\documentclass{article}

\usepackage{arxiv}

\usepackage[utf8]{inputenc} 
\usepackage[T1]{fontenc}    
\usepackage{hyperref}       
\usepackage{url}            
\usepackage{booktabs}       
\usepackage{amsfonts}       
\usepackage{nicefrac}       
\usepackage{microtype}      
\usepackage{lipsum}		
\usepackage{graphicx}
\usepackage{mathtools}
\usepackage{amssymb,amsmath,amsthm}

\title{Global exponential attitude tracking for spacecraft with gyro bias estimation}

\author{ \hspace{1mm}Eduardo Esp\'indola-L\'opez  \\
	Engineering Faculty\\
	National Autonomous University of Mexico\\
	Mexico City, MEXICO 3000\\
	\texttt{eespindola@comunidad.unam.mx} \\
	\And
	\hspace{1mm}Yu Tang \thanks{Corresponding author.} \\
	Engineering Faculty\\
	National Autonomous University of Mexico\\
	Mexico City, MEXICO 3000\\
	\texttt{tang@unam.mx} \\
}

\newtheorem{theorem}{Theorem}
\newtheorem{lemma}[theorem]{Lemma}
\theoremstyle{remark}
\newtheorem*{rmk}{Remark}
\hypersetup{
pdftitle={A template for the arxiv style},
pdfsubject={q-bio.NC, q-bio.QM},
pdfauthor={David S.~Hippocampus, Elias D.~Striatum},
pdfkeywords={First keyword, Second keyword, More},
}

\begin{document}
\maketitle

\begin{abstract}
This paper addresses the global exponential attitude tracking of a spacecraft when gyro measurements are corrupted by bias. Based on contraction analysis,  an exponentially convergent nonlinear observer is designed first to estimate the gyro bias. Relying on this bias estimator and the quaternion logarithm representation of the tracking error, an exponentially globally convergent controller is devised. This controller stabilizes the  unique equilibrium of the closed-loop system, where the tracking error is the unit quaternion. For more energy-efficiency and enhancing the robustness in the presence of measurement noise,  a hysteretically switching variable as in \cite{mayhew2011quaternion} is incorporated in the control loop and an unwinding-free globally exponentially convergent tracking controller is obtained. Numeric simulations were done to evaluate its performance in terms of tracking errors and energy-efficiency, as well as the robustness  to measurement noise and time-varying bias in gyro sensors.
\end{abstract}


\section{Introduction}
In spacecraft applications such as surveillance mapping, communication, deep space data acquisition and formation flying, accurate attitude control must be ensured. This topic has been studied extensively in the literature \cite{wen1991attitude,tayebi2008unit,mayhew2011quaternion,su2011globally,schlanbusch2012hybrid,gui2016global}. One major challenge for attitude control designs is the nonlinear relationship between the attitude representation and the angular velocity in the rotational kinematics, regardless the representation chosen to parameterize the attitude. Among commonly used attitude representations, unit quaternions are often preferred for being singularity-free compared with any three-parameter attitude representations (Euler angels, Rodriguez parameters and modified Rodriguez parameters) and easier to maintain its norm constraint than in a rotational matrix \cite{shuster1993survey}.  

Nevertheless, unit quaternions have an ambiguity: two unit quaternions ($\pm q$) correspond to the same rotation matrix and therefore the same physical attitude. This fact brings mainly two obstacles for controller designs: firstly, in a quaternion based control system there are two equilibria with the same desired attitude, stabilizing one of them would destabilize the other. Therefore, for a specific attitude trajectory arbitrarily close to the "unstable" equilibrium, the spacecraft will develop a full unnecessary rotation, causing the unwinding phenomenon. Secondly, achieving a global result is challenging since it requires to break the topological constraint and any continuous controller cannot attain this task \cite{bhat2000topological}.

In the past decades the unwinding phenomenon and global stabilizing using quaternions have been widely studied \cite{li2010global,mayhew2011quaternion,su2011globally,gui2016global}. A PID hybrid controller for global attitude tracking was designed in \cite{su2011globally}, global set stabilization using an optimal attitude controller can be found in \cite{li2010global}. Taking into practical issues such as available measurements in a low-cost application,  \cite{mayhew2011quaternion,gui2016global} considered the problem of designing a globally convergent unwinding-free tracking controller addressing gyro bias and noise in sensors.

Other ways for dealing with the angular velocity noise, bias and scaling factors are using nonlinear observers or filters \cite{salcudean1991globally,mahony2008nonlinear,martin2010design,chavez2018contracting}. \cite{mahony2008nonlinear} proposed non-linear filters to estimate the attitude and gyro bias from a low-cost IMU, \cite{martin2010design} proved the validity of this sort of nonlinear observers through a low-cost hardware implementation. Since the separation principle is not held for general nonlinear systems, incorporating an  observer into the closed loop control needs to redesign the observer or to establish the validity of the separation principle for a specific design. In this regards \cite{caccavale1999output,costic2001quaternion,wong2001adaptive,tayebi2008unit,zou2013finite} showed asymptotic stability for the proposed control using only attitude measurements. Furthermore, gyro bias correction was addressed in  \cite{thienel2003coupled,mayhew2011quaternion} to design an attitude feedback controller with asymptotic convergence.

Needless to say, energy-efficient control design is a critical issue for spacecraft applications. For asymptotic convergent  controller designs it has been noted  that some are more energy-efficient than others, being the main difference the way how to reach the desired attitude trajectory from an initial condition, determined by the tools used in the controller development. Passivity based designs \cite{egeland1994passivity,lizarralde1996attitude} have the main feature of modularity which simplifies significantly the overall design by devising   the controller, observer or adaptation into each functional module. However asymptotic convergence is commonly obtained. Lyapunov stability based designs have achieved a  stronger exponential convergence \cite{lee2012exponential,lee2015global,liu2017robust,xiao2019attitude}. But searching an "appropriate" Lyapunov function is not a trivial problem. Mostly,  auxiliary tools such as La'Salle invariance principle or Barbalat's Lemma must be invoked in order to reach the asymptotic result. 

Contraction analysis  \cite{lohmiller1998contraction,pavlov2006uniform,simpson2014contraction}, studied earlier in the mathematics literature (see \cite{jouffroy2005some} for a historical view) has been emerged as powerful alternative tools to design exponentially convergent observers and controllers.  It is closely related with the incremental  Lyapunov stability \cite{angeli2002lyapunov,forni2013differential}.  Contraction can be differentiated from Lyapunov stability by the convergence notion, that is, while Lyapunov analysis leads to the convergence to an equilibrium, contraction  analysis enables to conclude the convergence of any couple of trajectories \cite{lohmiller1998contraction,jouffroy2004methodological}. Further, the convergence of specific properties can be studied through the \textit{partial contraction} which is an extension of the contraction \cite{wang2005partial,russo2011symmetries} and provides a general framework to study the stability of nonlinear systems. Using such a  concept, contraction based  design consists mainly in two steps: first a "virtual system" is proposed verifying that the trajectories of interest, e.g., the target trajectory and that of the actual system, are its particular solutions; subsequently in the second step, contraction of the virtual system is shown ensuring that the distance, measured in an appropriate metric,  between the target trajectory and the actual trajectory decreases exponentially (contracting). More details can be found in \cite{lohmiller1998contraction,wang2005partial,jouffroy2004methodological}. 

Under piece-wise smooth (PWS) controls, the resulting closed-loop system  is also PWS and its solutions are defined in the sense of Filippov and are right-uniqueness \cite{goebel2009hybrid}.  Contraction analysis has been  extended to PWS systems in \cite{slotine2004contraction,rifai2006compositional,di2014contraction,fiore2016contraction}. 

By exploring  and leveraging the salient features provided by  contraction analysis, this paper considers the problem of design a global exponentially convergent attitude tracking controller,  addressing issues such as gyro bias as well as the topological constraint in attitude control with quaternion parameterization. 
The main contributions are stated as follows: first, a nonlinear gyro bias observer with global exponential convergence is designed based on quaternion kinematics. Second, relying on this bias observer and the quaternion logarithm of the tracking error an attitude tracking controller with global exponential convergence is developed. Third, a global exponential unwinding-free version of the previous controller is developed, where a hysteretically switching variable motivated by the hybrid control in \cite{mayhew2011quaternion} is incorporated into the controller.  

The rest of paper consists in seven sections. Section II gives the preliminaries, including contraction tools tailored to the subsequent designs, the rotational kinematics and dynamics of a spacecraft, the control objectives and the representation of the gyro measurements. In Section III the gyro bias observer is designed based on the spacecraft kinematics. Section IV develops the attitude tracking controller with the gyro bias correction.  Section  V devises   the global exponentially convergent switching controller. In Section  VI  numerical simulations were shown to evaluate the performance of the proposed controller in terms of tracking error and energy-efficiency, as well as the robustness in the presence of measurement noise and time-varying bias in the gyro sensor.  Finally, conclusions are given in Section VII.

\section{Problem Formulation and Definitions}

\subsection{Summary of Contraction Tools}
Using differential analysis, contraction analysis focuses on the convergence between any pair of trajectories of a given system (for a recent formal tutorial the reader is refereed to \cite{simpson2014contraction}). More precisely, consider the dynamical system 
\begin{equation} \label{eq:NonLS}
\dot{x}=f(x,t), \quad x(t_{0})=x_{0}, \quad \forall t\in \mathbb{R}_+,
\end{equation}
where the vector field $f:\mathbb{R}^n\times \mathbb{R}_+ \rightarrow \mathbb{R}^n$ is continuously  differentiable. 
Completeness of $f$ is assumed, i.e., for an initial condition $x_0$, the solution $\phi(x_0,t)$ exists $\forall t\geq t_0$.  Let $\mathcal X\subseteq \mathbb{R}^n$ a set in $\mathbb{R}^n$. The system \eqref{eq:NonLS} is said to be {\it contracting} in $\mathcal X$ with respect to the metric $\mathcal M={\Theta}^T{\Theta}$, where $\Theta$ is an $n\times n$ invertible   matrix, if there exists some $\lambda>0$ such that $\forall x\in \mathcal X$ and $\forall t\geq t_0$
\begin{equation} \label{eq:RegContrac}
 J^{T}(x,t)\mathcal{M} + \mathcal{M}J(x,t) \leq -2\lambda \mathcal{M},
\end{equation}
i.e., the symmetric part of the generalized Jacobian $J_G(x,t)$ $=$ ${\Theta}J (x,t){\Theta}^{-1}\leq -\lambda I$, $\forall x\in \mathcal X$, uniformly, where $I$ is the identity matrix of appropriate dimension and $J(x,t)$ $=$ $\frac{\partial f}{\partial x}$ $(x,t)$ is the Jacobian of \eqref{eq:NonLS}. 
$\mathcal X$ is called the {\it contraction region} and $\lambda$ is the {\it contraction rate}. If $\mathcal X=\mathbb{R}^n$, the contraction is global. Note that $\mathcal X$ defined in this way is convex and forward invariant for the system \eqref{eq:NonLS}.

In fact, consider the differential dynamics of \eqref{eq:NonLS}
\begin{equation} \label{eq:VarDin}
\delta \dot{x}=J(x,t)\delta x.
\end{equation}     
Let $V:=\delta x^{T}\mathcal{M}\delta x$ be the squared distance under the metric $\mathcal{M}$  between any pair of trajectories in $\mathcal X$.  The  time derivative of $V$ is 
\begin{equation*}
\dot{V} =\delta x^{T} ( J^{T}(x,t)\mathcal{M}  + \mathcal{M}J(x,t) ) \delta x. 
\end{equation*}
If the system \eqref{eq:NonLS} is contracting, by \eqref{eq:RegContrac}
\begin{equation}
\dot{V}(t) \leq -2\lambda V(t),
 \label{eq:TimeEvol}
\end{equation}
therefore $V(t)\leq V( t_0 ) e^{-2\lambda (t-t_0)}\; \forall  t\geq t_0$ and the exponential convergence of $\delta x(t)$ to zero    follows. Theorem \ref{thm1}  \cite{lohmiller1998contraction,wang2005partial} below  formalizes this fact. A concise formal  proof for a general state-dependent metric can be found in  \cite{aghannan2003intrinsic,simpson2014contraction}.

\begin{theorem}\label{thm1}
\emph{\textbf{(Contraction):}}
Consider the system \eqref{eq:NonLS}. Assume  the flow associate to $f$ to be forward complete. Then under the condition \eqref{eq:RegContrac} any pair of solutions of \eqref{eq:NonLS} $x(t)=\phi(x_0,t)$ and $y(t)$ $=$ $\phi$ $(y_0$ $,$ $t)$ with initial conditions $x_0, \ y_0\in \mathcal X$ will remain in $\mathcal X$ and 
\begin{equation} \label{eq:convergence}
    \| x(t)-y(t)\|\leq \| x_0-y_0\|e^{-\lambda (t-t_0)}, \ \forall t\geq t_0.
\end{equation}
\end{theorem}

To facilitate contraction based designs,  partial contraction is introduced \cite{wang2005partial}. The system \eqref{eq:NonLS} is said to be of {\it partial contraction} if its trajectories converge to the trajectories of a contracting system. This is made clear in the following theorem.

\begin{theorem}\label{thm2}
\emph{\textbf{(Partial contraction):}} 
Consider the following auxiliary system, termed virtual system
\begin{equation}\label{eq:VirtualSyst}
\dot{\xi}=\bar{f}\left( \xi ,x,t\right),
\end{equation}
associated with the nonlinear system \eqref{eq:NonLS} through $\bar{f}\left( x ,x,t\right)$ $=$ $f(x,t)$. Suppose the virtual system \eqref{eq:VirtualSyst} is contracting in $\xi$, $\forall \xi, \ x \in \mathcal X$, $t \geq t_{0}$. Then,  all its particular solutions converge exponentially to each other and in particular $(\xi - x)\to 0$ exponentially from any initial condition in $\mathcal X$. The system \eqref{eq:NonLS} is said to be {\it partially contracting}.
\end{theorem} 

Theorem \ref{thm2} enables the user to propose a virtual system with particular solutions the trajectories of the underlined systems (e.g., the trajectory of the target system and the trajectory of the actual system, both  initialized in a contraction region). Then contraction  of the virtual system implies the exponential convergence between these trajectories. 

Given tow contracting systems under possible different metrics, the cascade connection of these systems is also contracting provided that the  connection term is bounded \cite{lohmiller1998contraction,wang2005partial,simpson2014contraction}, as stated by the following theorem. 

\begin{theorem}\label{thm3}
\emph{\textbf{(Contraction of hierarchical systems):}}
Let two systems of possibly different dimensions
\begin{eqnarray} \label{eq:HierarcSys}
\dot{x}_{1} &=& f_{1}\left( x_{1},t \right), \nonumber\\
\dot{x}_{2} &=& f_{2}\left( x_{1},x_{2},t \right).
\end{eqnarray}
Consider the differential dynamics $\left[\delta x^{T}_{1} , \delta x^{T}_{2}\right]^T$, arranged as
\begin{equation}\label{eq:HierarcSys2}
\frac{d}{dt}
\left[
\begin{array}{c}
\delta x_{1} \\
\delta x_{2}
\end{array}
\right]
=\left[
\begin{array}{cc}
F_{1} & 0 \\
F_{21} & F_{2}
\end{array}
\right]
\left[
\begin{array}{c}
\delta x_{1} \\
\delta x_{2}
\end{array} 
\right] .
\end{equation}
If in some region of the state space $F_{1} := \frac{\partial f_{1}}{\partial x_{1}}$ and $F_{2} := \frac{\partial f_2}{\partial x_2}$ are uniformly negative definite, and $F_{21} := \frac{\partial f_2}{\partial x_1}$ is bounded, then the whole system \eqref{eq:HierarcSys} will be contracting in that region.
\end{theorem} 

System \eqref{eq:NonLS} may also represent the closed-loop dynamics of a controlled system with state feedback $u(x, t)$. For piece-wise smooth (PWS) $u(x,t)$, \eqref{eq:NonLS} holds with the right derivative at points of discontinuity, and its solutions are defined in the sense of Filippov and are right-uniqueness  \cite{lohmiller2000nonlinear,di2016switching}. Contraction analysis was extended to PWS systems in \cite{lohmiller2000nonlinear,di2014contraction,fiore2016contraction}. In particular, for continuous-time switching systems \footnote{With a little abuse of notation, $h$ is used here to identify two different continuous individual systems for $h \in \lbrace 1, -1\rbrace$.} \cite{goebel2009hybrid}
 \begin{align} \label{eq:hrbrid}
    \dot x(t) &=f_h(x(t),t), \quad h^{+}(t)=g(x(t),h(t),t),
\end{align}   
where $h\in\{-1,1\}$ is the discrete state, $f_h$ is the vector field for the two individual systems, $g$ is a function defining the switching rule, and $x\in \mathbb{R}^{n}$ is the continuous state. The following theorem, adopted  from \cite{slotine2004contraction} ( Theorem 5) and \cite{di2014contraction} ( Theorem 3.2), gives a sufficient conditions for contraction of switching  systems.

\begin{theorem}\label{thm4}
\emph{\textbf{(Contraction of switching systems):}} 
A continuous time switching system \eqref{eq:hrbrid} is contracting if the individual systems $f_h$, $h\in \{-1,\ 1\}$, are contracting with respect to a common metric $\mathcal{M}= \Theta^{T} \Theta$. 
\end{theorem} 

Under the condition that each individual system is contracting under a common metric, partial contraction stated in Theorem \ref{thm2} holds for the switching system \eqref{eq:hrbrid} \cite{di2014contraction} (Theorem 3.2).

\subsection{Spacecraft Rotational Dynamics}
The attitude of a spacecraft, denoted by a rotation matrix $R\in SO(3)$ defines  the orientation of the body reference frame $\mathbf{B}$ fixed to the mass  center of spacecraft respect to the inertial reference frame $\mathbf{I}$ fixed to the  center of the Earth. 

The parameterization of a rotation matrix by a unit quaternion is
\begin{equation*}
q=\left[ q_0 , q^{T}_{v} \right]^{T} \in \mathcal{S}^{3}, \; q_{0} \in \mathbb{R}, \; q_{v}=[q_1,\ q_2,\ q_3]^T
\in \mathbb{R}^{3},
\end{equation*}
where $\mathcal{S}^{3}=\lbrace x \in \mathbb{R}^{4} | x^{T}x=1 \rbrace$ represents a three dimension unit sphere embedded in $\mathbb{R}^{4}$. The corresponding  rotation matrix for  a given  quaternion is $R(q)= I + 2q_{0}S(q_{v}) + 2S^{2}(q_{v})$, with $I\in \mathbb{R}^{3\times 3}$ the identity matrix and $S(\cdot) \in \mathbb{R}^{3\times 3}$ the skew-symmetric operator 
\begin{equation*}
S(u)=
\left[
\begin{array}{ccc}
0 & -u_{3} & u_{2}\\
u_{3} & 0 & -u_{1} \\
-u_{2} & u_{1} & 0 
\end{array}
\right], \; \; u=\left[
\begin{array}{c}
u_{1}\\
u_{2} \\
u_{3} 
\end{array}
\right]\in \mathbb R^3.
\end{equation*}
Notice that $R(q)=R(-q)$, i.e., $q$ and $-q$ represent the same physical orientation.

The quaternions product $\otimes$ computes 
\begin{equation*}
q \otimes p = \left[
\begin{array}{c}
q_{0}p_{0} - q^{T}_{v}p_{v} \\
q_{0}p_{v} + p_{0}q_{v} + S(q_{v})p_{v}
\end{array}
\right].
\end{equation*}
Let  $q^{-1}=\left[q_0 , -q^{T}_{v}\right]^{T}$ be the conjugate of $q$, and $\hat{1}$ $=$ $\left[1,\right.$ $0,$ $0,$ $\left.0\right]^{T}$ the identity quaternion, then  $q\otimes q^{-1}=q^{-1}\otimes q = \hat{1}$.

The kinematics of a spacecraft
\begin{equation}\label{eq:KinematicsR}
\dot{R}(q) = R(q)S\left( \omega \right) ,
\end{equation}
where $\omega$ is the angular velocity in the body frame, expressed in terms of unit quaternions is
\begin{equation}\label{eq:KinematicsQ}
\dot{q}=\frac{1}{2}J(q)\omega ,
\end{equation}
where $J(q) \in \mathbb{R}^{4\times 3}$ given by
\begin{equation}\label{eq:MatJ}
J(q)=
\left[ 
\begin{array}{c}
-q^{T}_{v}  \\
q_{0} I + S(q_{v})
\end{array}
\right] :=
\left[ 
\begin{array}{c}
-q^{T}_{v}  \\
J_{v} (q)
\end{array}
\right] .
\end{equation}

\noindent \emph{\textbf{Properties of the matrix $J(x)$:}} \label{lema1}
For all $x,y \in \mathbb{R}^{4}$ the following properties of matrix $J(x)$ hold  \cite{markley2014fundamentals}:
\begin{enumerate}
\item $J^{T}(x)J(x)=||x||^{2}_{2} I$, \label{pA1}
\item $J^{T}(x)y = 0_{3\times 1} \iff y=kx, \quad k \in \mathbb{R}$, \label{pA2}
\item $J(\alpha x + \beta y)=\alpha J(x) + \beta J(y), \quad \alpha, \beta \in \mathbb{R} $,
\item $J^{T}(x)y = -J^{T}(y)x$,
\item $||J(x)||_{2}=||x||_{2}$,
\item $\frac{d}{dt}\left( J(x) \right) = J(\dot{x})$,
\end{enumerate}
where $0_{n\times m}$ is a matrix of $n\times m$ with zero in all its elements. 

The rotation dynamics of a spacecraft is given by
\begin{equation}\label{eq:Dynamics}
M\dot{\omega} = S(M\omega )\omega + \tau ,
\end{equation}
where $M \in \mathbb{R}^{3\times 3}$ is the constant inertia matrix $M =M^{T} >0$ and $\tau (t) \in \mathbb{R}^{3}$ is the torque control vector, both measured in the body frame.

\subsection{Measurements}\label{subsec:Meas}
In this paper, the attitude $q$ and the angular velocity $\omega$ are assumed to be available measurements. However, in practice the angular velocity from gyros are normally corrupted by bias and noise modelled as
\begin{equation*}
\omega_{g} = \omega + b + r_{g},
\end{equation*}
where $r_g \in \mathbb{R}^{3}$ represents the gyro noise and $b\in \mathbb{R}^{3}$ denotes the bias which is time-varying in the worst case. Nonetheless, to design the gyro bias observer, the measurements noise $r_g$ $=$ $0$ and constant bias are considered, the robustness to this assumption will be tested by simulations. Therefore, 
\begin{align}
\omega_{g} &= \omega + b, \label{eq:omIMU} \\
\dot{b}&=0. \label{eq:bias}
\end{align}
The estimated  angular velocity $\hat{\omega}\in \mathbb{R}^{3}$ is then  given by
\begin{equation}\label{eq:omEst}
\hat{\omega} = \omega_{g} - \hat{b},
\end{equation}
where $\hat{b}\in \mathbb{R}^{3}$ is the bias estimation.

\subsection{Tracking Error Dynamics and Control Objectives}
Given  a  smooth desired trajectory attitude trajectory $q_d (t)$ (with bounded $\dot q_d$) and the desired angular velocity  $\omega_d (t)$  related by $\dot{q}_d = \frac{1}{2}J(q_d )\omega_d$, define the attitude tracking error as
\begin{equation}\label{eq:eq}
e=q^{-1}_{d} \otimes q = \left[ e_{0} , e_{v}^{T} \right]^{T} =[e_0,\ e_1,\ e_2, \ e_3]^T \in \mathcal{S}^{3},
\end{equation}
its time derivative is then
\begin{equation}\label{eq:eqP}
\dot{e}=\frac{1}{2}J(e)\tilde{\omega},
\end{equation}
where $\tilde{\omega}\in \mathbb{R}^{3}$ is the angular velocity error
\begin{equation}\label{eq:OmTilde}
\tilde{\omega} = \omega - R^{T}(e)\omega_d. 
\end{equation}
Let 
\begin{equation}\label{eq:ln2}
z:= \frac{\theta}{2}\bar{k},
\end{equation} 
where $\left(\bar{k},\frac{\theta}{2}\right)$ is the axis/angle Euler representation of $R(e)$, with $\theta\in\mathbb{R}$, $\bar{k}\in \mathbb{R}^{3}$ and $\|\bar{k}\|=1$. Since $e=\big[ \cos(\frac{\theta}{2}),$ $\sin(\frac{\theta}{2})\bar{k}^{T} \big]^{T}$, the variable $z$ is related with the quaternion logarithm \cite{kim1996compact} $\ln : \mathcal{S}^{3} \to \mathbb{R}^{3}$  by 
\begin{equation}\label{eq:lnQ}
z(e) = \ln(e) = \arccos(e_0) \frac{e_{v}}{\|e_{v}\|}
\end{equation}
Note that $z(e)$ is continuously differentiable for all $t\geq 0$, and $z(e)=0\iff \theta=0\iff e_0 = 1 \implies \| e_v \| = 0$. The kinematics in terms of $z$ is
\begin{equation} \label{eq:KinematicsZ} 
\dot z = \frac{1}{2} G(z) \tilde{\omega},
\end{equation}
where  $G(z)$ is given by
\begin{equation}  \label{eq:Geq}
G(z) = I + S(z) + \frac{1}{\| z \|^2}\left( 1 - \|z\|\frac{\cos \|z\|}{\sin \|z\|} \right) S^2 (z)
\end{equation}
From \eqref{eq:OmTilde}  and by using the fact that $G(z)z=z$ it gives 
\begin{equation} \label{eq:dlnQ2} 
\dot z = -\lambda_{c} z + \frac{1}{2} G(z)(\omega - \omega_r),
\end{equation}
where  $\omega_r$ is defined as
\begin{align}
\omega_r &=  -2\lambda_{c} z + R^{T}(e)\omega_d, \label{eq:OmR}
\end{align}
with $\lambda_c>0$ being a design parameter. The control objective is to achieve $e(t)\to \hat{1}$ and $\omega (t)  \to \omega_d (t)$ exponentially by showing that $z(t)\to  0_{3\times 1}$ and $\omega(t) \to \omega_r(t)$ exponentially.

\section{Ideal Gyro Bias Observer}
In this section, a gyro bias observer is designed based on spacecraft kinematics \eqref{eq:KinematicsQ}. Observer of this section will serve as a baseline design for the later controller developments.

The ideal gyro bias observer is proposed as
\begin{align}
\hat{b} &= \bar{b} - K_o J^T (q_f)q, \label{eq:bE} \\
\dot{\bar{b}} &= \frac{1}{2}K_o J^{T}(q_f)J(q)\hat{\omega} + \gamma K_o J^{T}(q) q_f , \label{eq:bBar} \\
\dot{q}_f &= \gamma \left( q - q_f \right) ,\; q_{f}(0)=q(0), \label{eq:qfp}
\end{align}
where $K_o \in \mathbb{R}^{3\times 3}$,  $K_{o}=K_{o}^{T}>0$  is the observer gain, $\gamma >0$ is a constant gain of the first order filter \eqref{eq:qfp}. The following lemma states a useful property of the linear filter \eqref{eq:qfp}.

\begin{lemma}\label{lem2}
\emph{\textbf{(Linear filter):}}
 For any  $\epsilon_q >0$, there exists a $\gamma'$ such that if $\gamma > \gamma'$, then $\| q - q_f \| < \epsilon_q$ for all $t\geq 0$.
\end{lemma}
\begin{proof}
Let $p=\frac{d}{dt}$ the time derivative operator. From \eqref{eq:qfp}, $q_f (t)$ is rewritten as $q_f (t) = \frac{\gamma}{p+\gamma} q(t) $. Then
\begin{equation*}
\| q (t) - q_f (t) \| = \| \frac{p}{p+\gamma}q(t) \| \leq  \| \frac{p}{p+\gamma} \| < \epsilon_q,
\end{equation*}
provided that $\gamma > \gamma'  \vcentcolon = \omega_{max} / \epsilon_q $, where $\omega_{max}$ is the frequency limit beyond which $q\approx 0$.
\end{proof}

\begin{theorem}\label{thm5}
\emph{\textbf{(Global exponential convergence of the bias observer):}}
For a given observer gain $K_o>0$, chose $\gamma>0$  as in Lemma \ref{lem2} such that 
\begin{equation*}
\lambda_{o} \vcentcolon = \left[ \lambda_{min}(K_{o}) - \lambda_{max}(K_{o})\epsilon_q \right]>0.
\end{equation*}
Then the observer, defined by \eqref{eq:bE}-\eqref{eq:qfp} forces $\hat{b}(t) \to b$ exponentially with the convergence rate $\lambda_o$, $\forall \hat b(0)\in \mathcal X_b:=\mathbb{R}^3$.
\end{theorem}
\begin{proof}
The proof consists in three steps. First, the dynamics of the observer \eqref{eq:bE}-\eqref{eq:qfp} is written into a suitable (analysis) form, which allows in the second step  to propose a virtual system. The virtual system has as particular solutions the trajectory of the observer and the trajectory of the target dynamics \eqref{eq:bias}. The contraction analysis in the third step concludes the contraction of the virtual system. Therefore by Theorem \ref{thm2}, it follows the exponential convergence of the estimate to its actual one.
 
By taking the time derivative of \eqref{eq:bE}, then substituting \eqref{eq:bBar}, \eqref{eq:qfp} and \eqref{eq:KinematicsQ} and 
invoking the Properties of matrix $J(x)$  the analysis form of the observer dynamics can be written as
\begin{equation}\label{eq:bEp}
\dot{\hat{b}} = -\frac{1}{2}K_o J^{T}(q_f)J(q)\left( \hat{b} - b \right),
\end{equation}
which together with bias dynamics \eqref{eq:bias} suggests the following virtual system
\begin{equation}\label{eq:VsysO}
\dot{\xi} = -\frac{1}{2}K_o J^{T}(q_f)J(q)\left( \xi - b \right).
\end{equation}

The virtual system $\xi$ has two particular solutions: $\xi = \hat{b}$ and $\xi = b$ which corresponds to trajectory of the observer \eqref{eq:bEp} and trajectory of the bias dynamics \eqref{eq:bias}, respectively.

The differential  dynamics of \eqref{eq:VsysO} is
\begin{equation}\label{eq:VarDynObsI}
\delta\dot{\xi} = J_{o} \delta\xi,
\end{equation}
where $J_o$ is the Jacobian of the ideal observer
\begin{equation}\label{eq:JIObs}
J_{o} = -\frac{1}{2}K_o J^{T}(q_f)J(q).
\end{equation}

Let  the metric be $\mathcal{M}_1 = I$. The squared length of $\delta\xi$ under $\mathcal{M}_1$ is $V_1 = \delta\xi^{T} \mathcal{M}_1 \delta\xi = \delta\xi^{T} \delta\xi$. The time evolution of $V_1$ is
\begin{eqnarray}
\dot{V}_1 &=& 2\delta\xi^{T} \delta\dot{\xi} = 2\delta\xi^{T} J_o \delta\xi  \nonumber \\
        &=& -\delta\xi^{T}K_o J^{T}(q_f)J(q) \delta\xi \nonumber \\ 
        &=& -\delta\xi^{T} K_{o} \delta\xi - \delta\xi^{T} K_{o}J^{T}(q_f - q)J(q) \delta\xi \nonumber \\ 
        &\leq & - \lambda_{min}(K_{o})\| \delta\xi \|^{2} \nonumber \\
        & &+ \lambda_{max}(K_{o})\| J^{T}(q_f - q)\|\|J(q)\| \| \delta\xi \|^{2} \nonumber \\ 
        &\leq & - \lambda_{min}(K_{o})\| \delta\xi \|^{2} + \lambda_{max}(K_{o})\| q_f - q\| \| q\| \| \delta\xi \|^{2} \nonumber \\ 
 &\leq & - \left( \lambda_{min}(K_{o}) - \lambda_{max}(K_{o})\epsilon_q \right) \| \delta\xi \|^{2} \nonumber \\ 
 &= & - \lambda_{o} \| \delta\xi \|^{2}. \label{eq:TimeEvolObsI}
\end{eqnarray}
Therefore, for $\lambda_{o}>0$ the virtual system \eqref{eq:VsysO} is contracting in $\mathcal X_b$ and all its particular solutions converge exponentially to each other by Theorem \ref{thm2}. In particular, $\hat{b}(t)\to b$ exponentially with the convergence rate $\lambda_{o}$, $\forall \hat b(0)\in \mathcal X_b$.
\end{proof} 

\begin{rmk}[Filter]
The aims of filter \eqref{eq:qfp} are to filter out the measurement noise in quaternion for frequencies beyond  $\omega_{max}$ defined in Lemma \ref{lem2} and to get rid of using the time derivative of the quaternion in the observer implementation. 

For the case when $K_o$ $=$ $k_o I$ for some scalar $k_o$ $>$ $0$, the observer convergence rate becomes to $\lambda_{o} = k_o \left[ 1- \epsilon_q \right]$. If $0<\epsilon_q <1$, condition of Theorem \ref{thm5} is fulfilled, which is always possible by choosing  $\gamma >$ $ \omega_{max}/\epsilon_q$ according to Lemma \ref{lem2}.
\end{rmk}


\section{Tracking Controller with Gyro Bias Estimation}
Observer \eqref{eq:bE}-\eqref{eq:qfp} is used here to design an attitude control using the angular velocity estimation \eqref{eq:omEst}. The following tracking controller is proposed
\begin{eqnarray} \label{eq:OCtrl}
\tau &=& M \dot{\hat{\omega}}_r - S(M\hat{\omega})\omega_r - \frac{1}{2}G^{T}(z)z  \nonumber \\
& & - (K_c - 2\lambda_c P)(\hat \omega-\omega_r),
\end{eqnarray}
where $\lambda_c>0$ and and $0< K_c=K_c^T \in \mathbb{R}^{3\times 3}$ are the controller gains. The matrix $P_{a} =  \frac{1}{2}(P-P^{T})$ is the skew-symmetric part of matrix $P := M G(z)$. The estimated angular velocity $\hat{\omega}$, the quaternion error $e$, the reference angular velocity $\omega_r$ and the function $G(z)$  are defined in \eqref{eq:omEst}, \eqref{eq:eq}, \eqref{eq:OmR} and \eqref{eq:Geq} respectively. In addition, $\dot{\hat{\omega}}_r$ denotes $\dot {\omega}_r$ with ${\hat \omega}$  replacing by ${\omega}$ and is calculated by
\begin{eqnarray}
\dot{\hat{\omega}}_r &=& 2\lambda_{c}^{2}z+\lambda_{c} G(z)\omega_{r} +R^{T}(e)\dot{\omega}_{d} \nonumber \\ 
& & -\left( \lambda_{c} G(z) - S( R^{T}(e)\omega_d ) \right)\hat{\omega} .\label{eq:OmRpE_1}
\end{eqnarray}

To ensure the overall system to be contracting, the ideal observer \eqref{eq:bE}-\eqref{eq:bBar} is modified to 
\begin{align}
\hat{b} &= \bar{b} - K_o J^T (q_f)q - 2\lambda_c M z, \label{eq:bBarCO}\\ 
\dot{\bar{b}} &= \frac{1}{2}K_o J^{T}(q_f)J(q)\hat{\omega} + \gamma K_o J^{T}(q) q_f - 2\lambda^{2}_c M z, \notag 
\end{align}
where $q_f$ is the filtered $q$ in \eqref{eq:qfp}.

\begin{theorem}\label{thm6} 
\emph{\textbf{(Global exponential convergence of the continuous attitude controller):}}
The control law \eqref{eq:OCtrl}-\eqref{eq:OmRpE_1} together with the observer 
\eqref{eq:bBarCO} in closed-loop with the system \eqref{eq:KinematicsQ} and \eqref{eq:Dynamics} drives $\hat{b}$ $\to$ $b$, $\omega$ $\to$ $\omega_d$ and $z \to 0_{3\time 1}$ exponentially from any initial condition $[\hat b^T(0), \ \omega^T (0), \ z(0)^T  ]^T \in \mathcal X_c:=\mathbb{R}^3\times \mathbb{R}^3\times \mathbb{R}^3$. Consequently, $q\to q_d$ exponentially $\forall q(0)\in \mathcal{S}^{3}. $
\end{theorem}
\begin{proof}
As in the proof of Theorem \ref{thm5}, the proof here is carried out in three steps. First, the observer dynamics and controller are rewritten in their analysis form as follows. 

Observer dynamics is calculated from \eqref{eq:bBarCO} as
\begin{align*}
\dot{\hat{b}} =& \dot{\bar{b}} - K_o J^T (q_f)\dot{q} - K_o J^T(\dot{q}_f)q - 2\lambda_c M \dot z  \\
 =&\frac{1}{2}K_o J^{T}(q_f)J(q)\hat{\omega} + \gamma K_o J^{T}(q) q_f - 2\lambda^{2}_c M z  \\  
 &  - K_o J^T (q_f)\frac{1}{2}J(q)\omega - K_o J^T\big(\gamma ( q - q_f )\big) q  \\ 
 &   - 2\lambda_c M \left[ -\lambda_{c} z + \frac{1}{2} G(z) \big(\omega - \omega_r\big) \right]  \\ 
 =& \frac{1}{2}K_o J^{T}(q_f)J(q)\big( \hat{\omega} - \omega \big) - \lambda_c M G(z)\big(\omega - \omega_r\big),
\end{align*}
where \eqref{eq:KinematicsQ}, \eqref{eq:qfp}, \eqref{eq:bBarCO}, \eqref{eq:dlnQ2} and Properties of matrix $J(x)$  were used. Then substituting \eqref{eq:omIMU} and \eqref{eq:omEst} gets
\begin{align}
\dot{\hat{b}} &= -\frac{1}{2}K_o J^{T}(q_f)J(q)\left( \hat{b} - b \right) - \lambda_c P (\omega - \omega_r) \notag\\
:&=\bar f_1(\bar x,x), \label{eq:bEcoP}
\end{align}
with $x^T=\left[x_1^T,\right.$ $x_2^T,$ $ \left. x_3^T \right]$ $:=$ $ \left[ b^T, \right.$ $ \omega^T,$ $\left. z_d^T\right]$ 
and $\bar x^T = [\bar{x}_1^T,$ $ \  \bar{x}_2^T, \ \bar{x}_3^T]:=\left[\hat b^T, \ \omega^{T}_{r}, \ z^T \right]$.

The analytical form of the controller \eqref{eq:OCtrl} is obtained by first taking the time derivative of \eqref{eq:OmR}
\begin{equation*}
\dot{\omega}_r = \dot{\hat{\omega}}_r - \left( \lambda_{c} G(z) - S( R^{T}(e)\omega_d ) \right)( \omega - \hat{\omega} ),
\end{equation*}
and then pre-multiplying  $M$ in both sides and substituting the controller \eqref{eq:OCtrl} as
\begin{align}
M\dot{\omega}_r =& S(M\omega ) \omega_r + \left( K_c -2\lambda_c P_{a} \right) \left( \omega - \omega_r \right) \nonumber  \\ 
 &+ \big( F_r - \lambda_c P^{T} \big) \big( \hat{b} -b \big) + \frac{1}{2} G^{T}(z) z  + \tau \notag\\
 :=& \bar f_2(\bar x,x)+\tau,   \label{eq:CtrlAnyOF}
\end{align}
where $F_r \in \mathbb{R}^{3\times 3}$ is given by
\begin{equation}\label{eq:MatFOF}
F_r =  S(\omega_r ) M - K_c + M S\big( R^{T}(e)\omega_d \big),
\end{equation}
which is bounded for bounded $\omega_d$. 

By the last, the analytical form of the overall system is completed by recalling that the logarithm error quaternion in \eqref{eq:dlnQ2}
\begin{align} \label{eq:dlnQ}
\dot z &=-\lambda_{c} z  + \frac{1}{2} G(z)(\omega - \omega_r)
:=\bar f_3(\bar x,x).
\end{align}

Notice that \eqref{eq:Dynamics}, \eqref{eq:bias} and the desired  dynamics of quaternion logarithm  $\dot{z}_d=-\lambda_c z_d$ represent the current dynamical system:
\begin{align}
    \dot b&=0:=f_1(x), \label{overal-dy1}\\
   M\dot{\omega} &= S(M\omega )\omega + \tau:=f_2(x)+\tau, \label{overal-dy2}\\
   \dot z_d&=-\lambda_c z_d:=f_3(x), \label{overal-dy3}
\end{align}
which suggests the following virtual system
\begin{align}
    \mathcal M_2 \dot\xi=\bar f(\xi,x)+\bar{\tau}, \label{eq:VsysOF}
\end{align}
where $\mathcal{M}_2 \vcentcolon = \textsl{diag} \lbrace I ,M ,I \rbrace>0 $, $\xi^T $ $=$ $[ \xi^{T}_{1}$ $ , \xi^{T}_{2}$ $ ,\xi^{T}_{3}]$ , $\bar{\tau}^T$ $=$ $[0_{3\times 1}^T $ $, \tau^T $ $, 0_{3\times 1}^T]$ and $\bar f^T(\xi,x)$ $=$ $[\bar f_1^T(\xi,x)$ $ , \bar f_2^T(\xi,x) $ $, \bar f_3^T(\xi,x)]$ associated with the system dynamics \eqref{overal-dy1}-\eqref{overal-dy3}   through $\bar f(x,x)=f(x)$, defined in $\mathcal X_c:=\mathbb{R}^3\times \mathbb{R}^3\times \mathbb{R}^3$.

Clearly this virtual system $\xi$ has two particular solutions: 

\noindent $\left[ \hat{b}^{T},\ \omega^{T}_{r}, \ z^T \right]^{T}$ and $\left[ b^{T}, \ \omega^{T}, \ z_d^T  \right]^{T}$, corresponding to the trajectory of the closed-system \eqref{eq:bEcoP}-\eqref{eq:dlnQ} and those of \eqref{overal-dy1}-\eqref{overal-dy3}, respectively. 
To analyze the contraction property of the  virtual system \eqref{eq:VsysOF}, calculate its differential dynamics
\begin{equation}\label{eq:VarDynOF}
\mathcal{M}_2 \delta \dot{\xi} = J_{oc_1} \delta \xi,
\end{equation}
where $J_{oc}$ is the  Jacobian of the observer/controller dynamics
\begin{align}
J_{oc} &= 
\left[
\begin{array}{ccc}
 J_{o} & 0_{3 \times 3} & 0_{3 \times 3} \\
 F_r & -K_c & 0_{3 \times 3} \\
 0_{3 \times 3} & 0_{3 \times 3} & -\lambda_c I
\end{array} 
\right] \nonumber \\
+&\left[
\begin{array}{ccc}
 0_{3 \times 3} & \lambda_c P & 0_{3 \times 3} \\
 -\lambda_c P^{T} & S\left( M \omega \right)+2\lambda_c P_{a} & \frac{1}{2} G^{T}(z) \\
 0_{3 \times 3} & -\frac{1}{2} G(z) & 0_{3 \times 3}
\end{array} 
\right],  \label{eq:Jacob_OF}
\end{align}
with $J_o$ the Jacobian of the ideal observer defined in \eqref{eq:JIObs}.

Taking $\mathcal M_2$ as the metric, let $V_2 = \delta\xi^{T} \mathcal{M}_2 \delta\xi$ the squared length of $\delta\xi$ under the metric $\mathcal{M}_2$. Its  time evolution of is
\begin{equation}\label{eq:TimeEvol_OF}
\dot{V}_2 = 2\delta\xi^{T} \mathcal{M}_2 \delta \dot{\xi} = 2\delta\xi^{T} J_{oc} \delta \xi = 2\delta\xi^{T} J_{s} \delta \xi, 
\end{equation} 
with
\begin{align}\label{eq:Jacob_OFsym}
J_{s} &=
\left[
\begin{array}{cc}
 J_{o} & 0_{3 \times 6} \\
 F & J_{c}
\end{array}
\right], \notag\\
F &= \left[
\begin{array}{c}
 F_r \\
 0_{3 \times 3}
\end{array}
\right], \;
J_{c} = \left[
\begin{array}{cc}
 -K_c & 0_{3\times 3} \\
 0_{3 \times 3} & -\lambda_c I
\end{array}
\right].
\end{align}
The Jacobian $J_{s}$ has the hierarchical structure of \eqref{eq:HierarcSys2} with negative definite matrices $J_o$ and $J_{c}$. Given that $F_r$ is bounded, contraction of virtual system \eqref{eq:VsysOF} follows from Theorem \ref{thm3}. Therefore, all its particular solutions converge exponentially to each other, in particular $\hat{b}(t)\to b$, $\omega_r (t) \to \omega (t)$ and $z(t)\to z_d(t) \to 0$ exponentially for all initial conditions $[\hat b^T(0), \ \omega^T (0), \ z(0)^T  ]^T\in \mathcal X_c$, which in turn implies $e(t) \to {\hat{1}}$ and $\omega(t) \to \omega_d (t)$ exponentially.
\end{proof} 

\begin{rmk}[The continuous controller]
By using the quaternion logarithm \eqref{eq:eq}, $e=\hat 1$ is the {\it only equilibrium} of the closed-loop system. 
The continuous controller \eqref{eq:OCtrl}-\eqref{eq:OmRpE_1} stabilizes the equilibrium $e_0=1$ (i.e., $e=+\hat 1$) instead $e_v=0_{3\times 1}$ (i.e., $e=\pm\hat 1$) as in most reported works (see, e.g.,  \cite{mayhew2011quaternion, wen1991attitude,tayebi2008unit,thienel2003coupled} and the references therein). This enables the controller to achieve the 
{\it global} convergence of $e=+\hat 1$ instead the {\it almost global} convergence when stabilizing the point $e_v=0_{3\times 1}$ (i.e., $e=\pm \hat 1$). Stabilizing one of the {\it two  equilibria} $e=\pm \hat 1$ will leave the other to be an unstable equilibrium, creating the unwinding phenomenon. Unwinding not only wastes energy, but also may introduce instability in the presence of arbitrarily small perturbations under a discontinuous control \cite{mayhew2011quaternion}.  

Another salient feature of the proposed continuous controller is the stronger (exponential) convergence of $e=\hat 1$   instead of the asymptotic convergence in, for example,  \cite{mayhew2011quaternion, wen1991attitude,tayebi2008unit,thienel2003coupled}. 
\end{rmk}

\section{Unwinding-Free Attitude Tracking Observer-Controller Design}
Given that $e=+\hat{1}$ and $e = -\hat{1}$ represent the same physical orientation, both points should be options to stabilize. In fact, the best option is to stabilize  the closest point from a given initial condition. The controller \eqref{eq:OCtrl}-\eqref{eq:OmRpE_1} can only stabilize $e = +\hat{1}$, then, for an initial attitude condition starting close to $e = -\hat{1}$ the control law will produce an unnecessary full rotation, causing energy waste.

 Motivated by the hybrid control in \cite{mayhew2011quaternion}, in this section an unwinding-free tracking controller is designed.
Instead of stabilizing one of the two points $e=\pm \hat{1}$, this controller stabilizes $e = h \hat{1}$ from any initial condition, where the hysteretically switching variable $h$ is
\begin{align}
    \dot{h}(t) =0, \mathrm{when} \ x \in \mathcal{C} := \{ x \in \mathcal{X}_{c}| h e_0 \geq -\delta \}, \label{eq:hbar} 
\end{align}
\begin{equation*}
    h^{+}(t) = \widehat{\mathrm{sgn}} (e_0 (t)), \mathrm{when} \ x \in \mathcal D:=\{ x \in \mathcal{X}_{c} | he_0\leq -\delta \},
\end{equation*}
with $h(0)=\widehat{\mathrm{sgn}}(e_{0}(0))$, where $h^+(t)$ denotes $h(t)$ right after the switching,  and $\delta \geq 0$ is the half-width of the hysteresis. The function  $\widehat{\mathrm{sgn}}(e_0) \in \lbrace 1, -1 \rbrace$ is defined as
\begin{equation}\label{eq:SgnF}
\widehat{\mathrm{sgn}}(e_0)=
\left\lbrace
\begin{array}{cc}
 1, & \mathrm{if} \  e_0\geq 0, \\
-1, & \mathrm{if} \  e_0<0,
\end{array}
\right.,
\end{equation}
which  is displayed in Fig. \ref{fig:Histeresis} with the paths in the direction of the arrow.
\begin{figure}
	\begin{center}
		\includegraphics[trim = 0mm 0mm 0mm 0mm,scale=0.5]{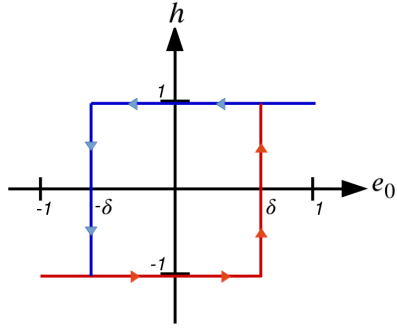}
		\caption{Graphical illustration for function $h\left( e_{0} \right)$.}
		\label{fig:Histeresis}
	\end{center}
\end{figure}
The  switching controller/observer is then obtained by replacing the quaternion error $e$ by $h e$ (notice that $z(e)$ is changed by $z(he)$) in the continuous controller/observer \eqref{eq:OCtrl}-\eqref{eq:bBarCO} 
\begin{eqnarray}
\tau &=& M \dot{\hat{\omega}}_{r} - S\left( M \hat{\omega} \right)\omega_r  - \frac{1}{2}G^{T} \left( z \right) z \nonumber \\ 
& & - \left(K_c - 2\lambda_c P_{a,h} \right)\left(\hat{\omega} - \omega_r\right), \label{eq:OFeedB_UF} \\
\omega_r &=&  -2\lambda_{c} z + R^{T}\left( h e \right)\omega_d, \label{eq:OmR_UF} \\
\dot{\hat{\omega}}_r &=& 2\lambda^{2}_{c} z + \lambda_{c} G \left( z \right)\omega_{r} + R^{T}( h e)\dot{\omega}_{d} \nonumber \\ 
& &  -\left( \lambda_{c} G \left( z \right) - S( R^{T}( h e )\omega_d ) \right)\hat{\omega} , \label{eq:OmRpE_UF} \\
z &=& z(h e), \label{eq:lnh} 
\end{eqnarray}
where $P_{a,h} = \frac{1}{2}\left(P_h - P^{T}_h \right)\in \mathbb{R}^{3\times 3}$ is the skew-symmetric part of $P_h = M G \left( z(he) \right)$ and the observer \eqref{eq:bBarCO} is used again to estimate $\hat{b}$ but replacing \eqref{eq:lnh}.

\begin{theorem}\label{thm7} 
\emph{\textbf{(Global exponential convergence of the attitude switched controller):}}
The control law \eqref{eq:OFeedB_UF}-\eqref{eq:lnh}
in closed-loop with the system defined by \eqref{eq:KinematicsQ}-\eqref{eq:Dynamics} drives $e \to h\hat{1}$, $\omega \to \omega_d$ and $\hat{b} \to b$ with $h$ defined by \eqref{eq:hbar}, exponentially from any initial condition.
\end{theorem}
\begin{proof}
The observer/controller in its analysis form is obtained as \eqref{eq:bEcoP}-\eqref{eq:dlnQ} in Theorem \ref{thm6} but using \eqref{eq:OFeedB_UF}-\eqref{eq:lnh} as follows
\begin{align}
\dot{\hat{b}} &= -\frac{1}{2}K_o J^{T}(q_f)J(q)\left( \hat{b} - b \right) - \lambda_c P_h (\omega - \omega_r) \nonumber\\
&:= \bar f_{1,h}(\bar x,x), \nonumber \\
M\dot{\omega}_r &= S(M\omega ) \omega_r + \left( K_c -2\lambda_c P_{a,h} \right) \left( \omega - \omega_r \right) \nonumber  \\ 
 & + \big( F_h - \lambda_c P^{T}_h \big) \big( \hat{b} -b \big) + \frac{1}{2} G^{T}(z) z   + \tau \notag  \\
 &:= \bar{f}_{2,h}(\bar x,x)+\tau, \nonumber  \\ 
\dot z &= -\lambda_{c} z + \frac{1}{2} G(z)(\omega - \omega_r) \nonumber\\
&:= \bar{f}_{3,h}(\bar x,x), \label{eq:OCanyU} 
\end{align}
with $x^T=\left[x_1^T,\right.$ $ x_2^T,$ $\left. x_3^T \right] $ $:=$ $ \left[ b^T ,\right.$ $ \omega^T ,$ $\left. z_d^T \right]$, $\bar x^T $ $=$ $\big[\bar{x}_1^T,$ $\bar{x}_2^T,$ $\bar{x}_3^T \big] $ $:=$ $\big[\hat b^T,$ $\omega^{T}_{r},$ $z^T \big]$, and $F_h \in \mathbb{R}^{3\times 3}$ is given by
\begin{equation}\label{eq:MatFUF}
F_h =  S(\omega_h ) M - K_c + M S( R^{T}(h e)\omega_d ),
\end{equation}
which is bounded for bounded $\omega_d$.

Thus, system \eqref{eq:OCanyU} suggests the following virtual system
\begin{align}
    \mathcal M_3 \dot\xi =\bar f_{h}(\xi ,x)+\bar{\tau}, \label{eq:VsysUF}
\end{align}
where $\mathcal{M}_3 \vcentcolon = \textsl{diag} \lbrace I ,M ,I \rbrace>0 $, $\xi^{T} = \left[ \xi^{T}_{1} , \xi^{T}_{2} , \xi^{T}_{3}\right]$ , $\bar{\tau}^T $ $=$ $\left[0_{3\times 1}^T , \tau^T , 0_{3\times 1}^T\right]$ and $\bar f_{h}^T(\xi ,x)=\big[\bar f_{1,h}^T(\xi ,x) $ $,$ $\bar f_{2,h}^T $ $(\xi ,x)$ $,$ $ \bar f_{3,h}^T(\xi ,x)\big]$ associated with $f_h$,  because $\bar f_h (x,x)=f_h(x) \in \mathcal X_c:=\mathbb{R}^3\times \mathbb{R}^3\times \mathbb{R}^3$,  where $f_h(x)=f(x)$ the spacecraft dynamics \eqref{overal-dy1}-\eqref{overal-dy3} which is independent of $h$.

The particular solutions of the virtual system are $[ \hat{b}^{T},$ $\omega^{T}_{r},\ z^{T}]^{T}$ and $[ b^{T}, \ \omega^{T}, \ z_d^{T}]$ corresponding to the trajectory of the closed-loop system dynamics and those of the spacecraft dynamics \eqref{overal-dy1}-\eqref{overal-dy3}, respectively.

The virtual system \eqref{eq:VsysUF} is a continuous-time switched system, with \eqref{eq:hbar} the discrete state, which provides the corresponding individual system  \cite{slotine2004contraction,di2014contraction}. Since $h \in \lbrace 1, -1\rbrace$, there are only two different vector fields. The differential dynamics of each continuous system is
\begin{equation}\label{eq:VarDynOF_UF}
\mathcal{M}_3 \delta \dot{\xi } = J_{oc,h} \delta \xi,
\end{equation}
where $J_{{oc},h}$ is
\begin{eqnarray}
J_{{oc},h} &=&
\left[
\begin{array}{ccc}
 J_{o} & 0_{3 \times 3} & 0_{3 \times 3} \\
 F_h & -K_c & 0_{3 \times 3} \\
 0_{3 \times 3} & 0_{3 \times 3} & - \lambda_{c} I
\end{array} 
\right] + \notag\\
& & \left[
\begin{array}{ccc}
 0_{3 \times 3} & \lambda_c P_h & 0_{3 \times 3} \\
 -\lambda_c P^{T}_h & S\left( M\omega \right)+2\lambda_c P_{a,h} & \frac{1}{2} G^{T} (z) \\
 0_{3 \times 3} & -\frac{1}{2} G (z) & 0_{3 \times 3}
\end{array} 
\right], \nonumber\\ \label{eq:Jacob_UF2}
\end{eqnarray}
with $J_o$ the Jacobian of the ideal observer defined in \eqref{eq:JIObs}.

Under  the constant metric $\mathcal{M}_3$ for both individual systems, the squared length of $\delta\xi$ measured by $V_3 = \delta\xi^{T} \mathcal{M}_3 \delta\xi$ has time derivative
\begin{equation}\label{eq:TimeEvolOF_UF}
\dot{V}_3 = 2\delta\xi^{T} J_{oc,h} \delta \xi = 2\delta\xi^{T} J_{s,h} \delta \xi , 
\end{equation}
where $J_{s,h}$ is
\begin{equation}\label{eq:Jacob_UFsym}
J_{s,{h}} =
\left[
\begin{array}{cc}
 J_{o} & 0_{3 \times 6} \\
 F_{2,h} & J_{c}
\end{array}
\right],
F_{2,h} = \left[
\begin{array}{c}
 F_h \\
 0_{3 \times 3}
\end{array}
\right],
\end{equation}
with $J_c$ the Jacobian of the ideal controller defined in \eqref{eq:Jacob_OFsym}.

The Jacobian $J_{s,h}$ has the hierarchical structure of \eqref{eq:HierarcSys2} with $J_o$, $J_{c}$ negative definite and $F_h$ is bounded. Contraction of individual continuous systems 
follows from Theorem \ref{thm3}. Moreover, since both individual systems share the same metric $\mathcal{M}_3$ for all $t$, contraction of the continuous-time switching system \eqref{eq:VsysUF} follows from Theorem \ref{thm4} and the comments after this theorem. As a consequence,  all particular solutions of the virtual system \eqref{eq:VsysUF} contracts to each other, in particular $\hat{b} \to b$, $\omega_r \to \omega$ and $z \to z_d\to 0_{3\times 1}$, that is, $h e \to \hat{1}$ and $\omega \to \omega_d$ exponentially from any initial condition in $\mathcal X_c$.
\end{proof}

\begin{rmk}[The switching controller]
Similar to that in \cite{mayhew2011quaternion}, the hysteretically switching variable $h$ \eqref{eq:hbar} aims at 
enabling the controller to stabilize the closest point $e=+ \hat{1}$ or $e=-\hat{1}$ according to the  half width  $\delta$ of the hysteresis, avoiding excessive waste of energy caused by full rotations. Also, it  makes controller robust to noisy measurements by replacing a single switching point (discontinuous control) by a switching region. The width of the switching region may be determined proportional to the noise size. 

The control law \eqref{eq:OFeedB_UF}-\eqref{eq:lnh} 
for $\delta = 0$, corresponds to the common discontinuous control case. On the other hand, for $\delta = 1$, it corresponds to the continuous control case, i.e., controller \eqref{eq:OCtrl}-\eqref{eq:OmRpE_1} and observer \eqref{eq:bBarCO}.

The number of switching is bounded for any closed-loop solution under the hysteretically switching control \eqref{eq:OFeedB_UF}-\eqref{eq:lnh} as established in  \cite{mayhew2011quaternion} (Theorem 5.3).
\end{rmk}

\section{Simulations}
Three simulations are carried out.  In  the first two simulations the  performance of proposed controller is evaluated in terms of tracking errors and energy-efficiency,  taking as a reference the controller reported in  \cite{thienel2003coupled} for both situations when $\delta = 1$ (continuous) and for $\delta = 0.3$ (unwinding-free).  In the third simulation, noise on gyro measurements and random bias are considered to evaluate the robustness of the  proposed controller.

The first simulation (continuous-control case) evaluates how the convergence (i.e., exponential vs asymptotic) affects the energy-efficiency. In this case, both the proposed controller and that of  \cite{thienel2003coupled}  stabilize the equilibrium  $e=+\hat{1}$ although the closed equilibrium is $e=-\hat{1}$ . The controller  gains were chosen in such a way that both controllers spend approximately the same energy for tracking the desired trajectory. Secondly, a hysteresis width $\delta = 0.3$ was incorporated in the proposed controller, illustrating its performance when unwinding phenomenon is eluded.
All parameters  and initial conditions were chosen the same as in \cite{mayhew2011quaternion}:  $M =\textsl{diag}\lbrace 10 \bar{u} \rbrace$ with $u=[1,2,3]^{T}$, $\bar{u}=u/\| u \|$, $h (0)=1$, $q(0)=[-0.2,\sqrt{1-(0.2)^{2}}\bar{u}]^{T}$ and $\omega (0)=0.5\bar{u}$. Additionally, the desired trajectory and true bias values were taken from \cite{thienel2003coupled} as $\dot{\omega}_d (t) = 0_{3\times 1}$, $\omega_d (t) = [0,0.11,0]^{T}$, $q_d (0)=[1,0,0,0]^{T}$ and $b=[0.05,-0.05,0.033]^{T}$. 

Controller gains  were  $K_c = 1.0 I$, $\lambda_c = 0.01$, $K_o = 1.0 I$ and $\gamma = 0.5$ for the proposed controller and  $k_D = 4.5 I$, $\lambda = 0.045$ and $k = 1.0$  for controller of \cite{thienel2003coupled}, seeking for spending  approximately the same energy in both controllers.

In Fig. \ref{fig:ID_Graf} both controller responses are displayed when $\delta = 1$. 
Fig. \ref{fig:ID_Graf}. e) verifies that both controllers are investing the approximately  same amount of energy with the selected gains through the control effort $\sqrt{\int_{0}^{t} \tau^{T}\tau dt}$. Under that condition, it can be observed from $e_{0}$ (Fig. \ref{fig:ID_Graf}. a)) that the proposed controller tracks the target $90$ $[s]$ faster than the controller of \cite{thienel2003coupled}, with the angle between the actual  and the desired attitude  through the graphic of $2\arccos |e_{0}|$ shown in Fig. \ref{fig:ID_Graf}. b). This was achieved at the price that the tracking error in the angular velocity is larger in the  transient  in Fig. \ref{fig:ID_Graf}. c). The closed loop system reached to the  steady-state at  about $24$ $[s]$,  where the proposed controller has a smaller angular velocity tracking error.   On the other hand,   observers performance measured by  $\| \hat{b} - b \|$  is shown in Fig. \ref{fig:ID_Graf}. d). Both  observers led to  the convergence of the bias error to zero  at a similar time with a better transient performance noticed in the proposed observer.  In consequence, the proposed observer is  more energy-efficient  than the controller in \cite{thienel2003coupled}. This is partly due to the stronger exponential convergence demonstrated in this paper instead of the asymptotic convergence  presented in \cite{thienel2003coupled}.

To evaluate further the energy saving when unwinding phenomenon is present, the proposed   controller  with $\delta = 0.3$ (unwinding-free controller) was used instead of the proposed continuous controller in the first simulation. Therefore, it stabilized $e=-\hat{1}$ instead of $e=+\hat{1}$.  Fig. \ref{fig:UW_Graf}. a) and b) show that the hysteresis function switches the control action to stabilize $e = -\hat{1}$ after $e_{{0}}$ reached at $-0.3$ which led  to the actual trajectory continuing its natural motion given the initial conditions. Besides, the angular velocity  tracking error $\| \tilde{\omega} \|$ and the bias estimation error (Fig. \ref{fig:UW_Graf} c) and d), respectively) are smaller in  the proposed controller with a similar convergence time as in the first simulation. More importantly, the profit obtained by avoiding the unwinding phenomenon is the reduced energy as evidenced by  Fig. \ref{fig:UW_Graf}. e), where about $30\%$ of energy saving compared with the controller in  \cite{thienel2003coupled} was observed.  

To test the robustness of the proposed controller, in the third simulation gyro measurements and bias were contaminated by random  noise 
\begin{align}
\omega_{g} &= \omega + b + r_g, \\
\dot{b}&=r_b.
\end{align}
where $r_g = m_1 \bar{\nu}$ and $r_b = m_2 \bar{\nu}$, with $\nu \in \mathbb{R}^{3}$, $\bar{\nu}=\nu/\| \nu \|$  a zero-mean Gaussian  white noise with the covariance $0.5$,  and $0\leq m_1\leq 0.01$ and $0\leq  m_2\leq 0.03$ are  scaling factors with uniform distribution. 

Fig. \ref{fig:R_Graf}. e) shows the time-varying bias and the estimated bias. Notice that the estimated bias followed  closely the random trajectory of the true bias. In Fig. \ref{fig:R_Graf}. a), the quaternion error $e$ is displayed, showing how unwinding phenomenon is elude whereas quaternion error oscillates very close to $-\hat{1}$. Fig. \ref{fig:R_Graf}. b) illustrates the tracking error in both the norm $\| \tilde{\omega} \|_2$ and  the root mean square (RMS) as a function of time, i.e., $\| \tilde{\omega} \|_{rms} = \sqrt{ \frac{1}{t}\int_{0}^{t} \| \tilde{\omega} \|^{2} dt}$. The graphic of $\| \tilde{\omega} \|_{rms}$ indicates that the tracking error $\tilde{\omega}$ remained  close to $0$ for all $t\geq 0$ since $\| \tilde{\omega} \|_{rms}$ is decreasing exponentially. Moreover, a similar behavior can be seen for $\| b - \hat{b} \|_{rms}$ in Fig. \ref{fig:R_Graf}. c). By last, Fig. \ref{fig:R_Graf}. d) shows the control effort which is practically the same as the case  when gyro bias was constant. 

\begin{figure}
	\begin{center}
		\includegraphics[scale=0.3]{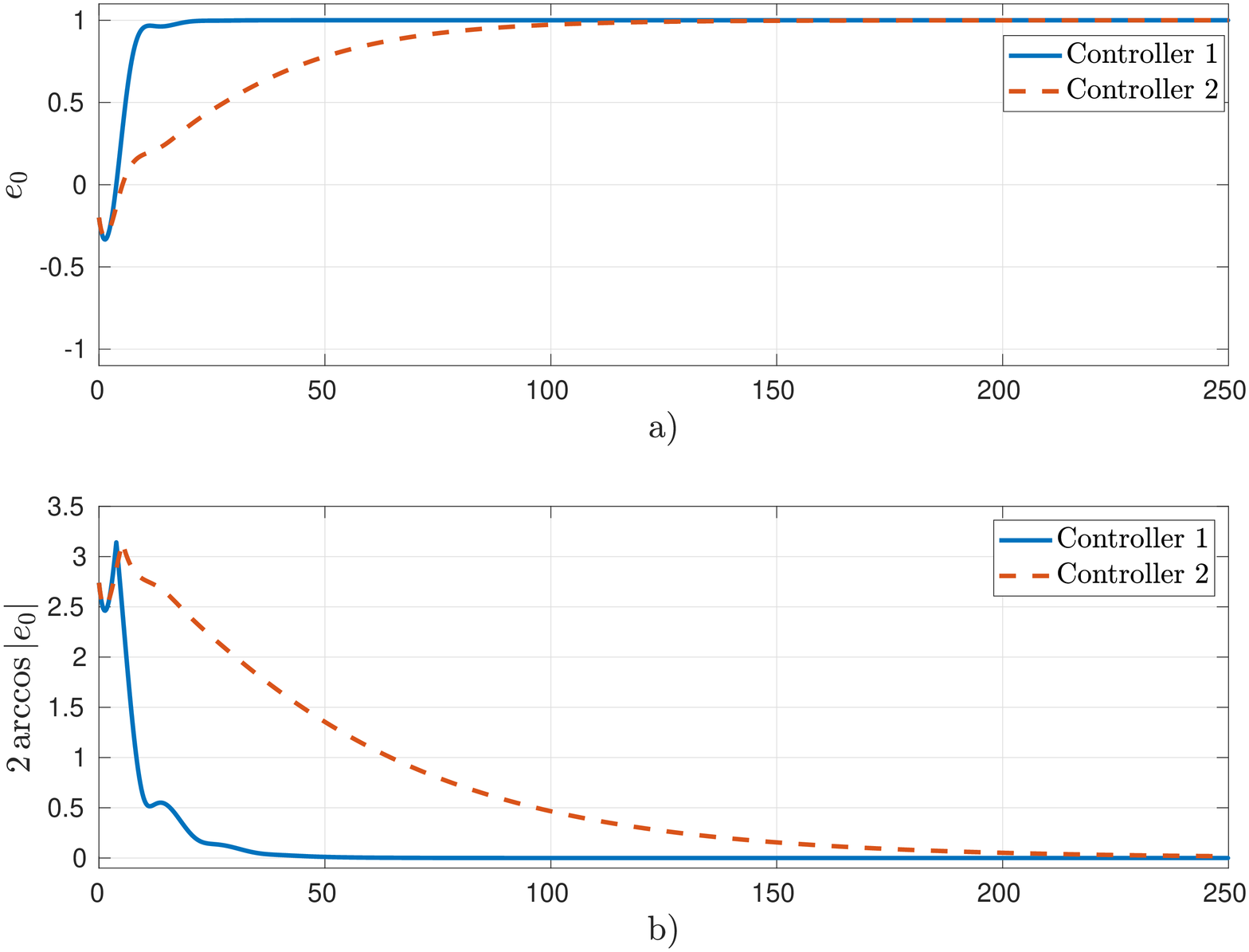}
		\includegraphics[scale=0.3]{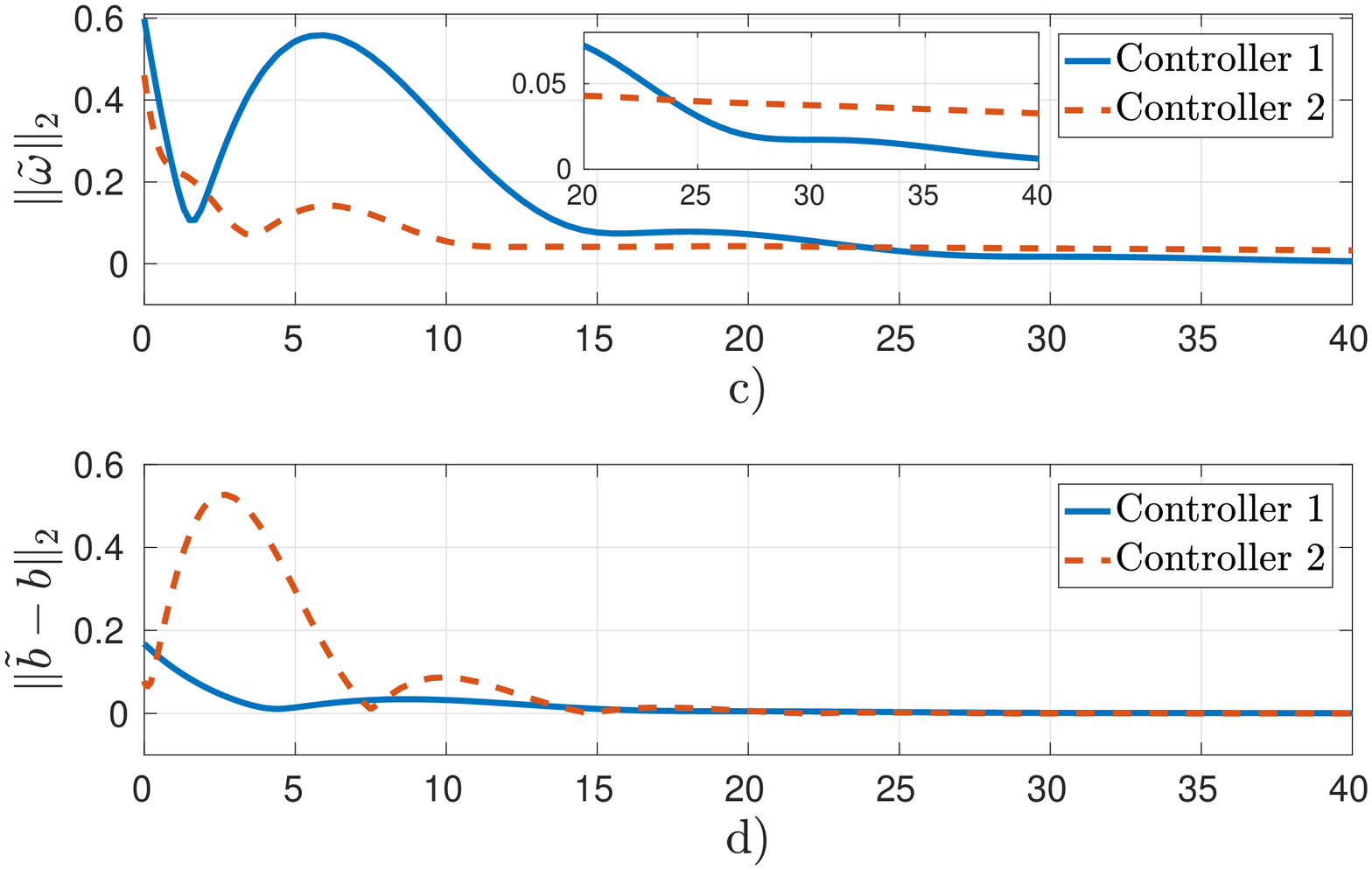}
		\includegraphics[scale=0.3]{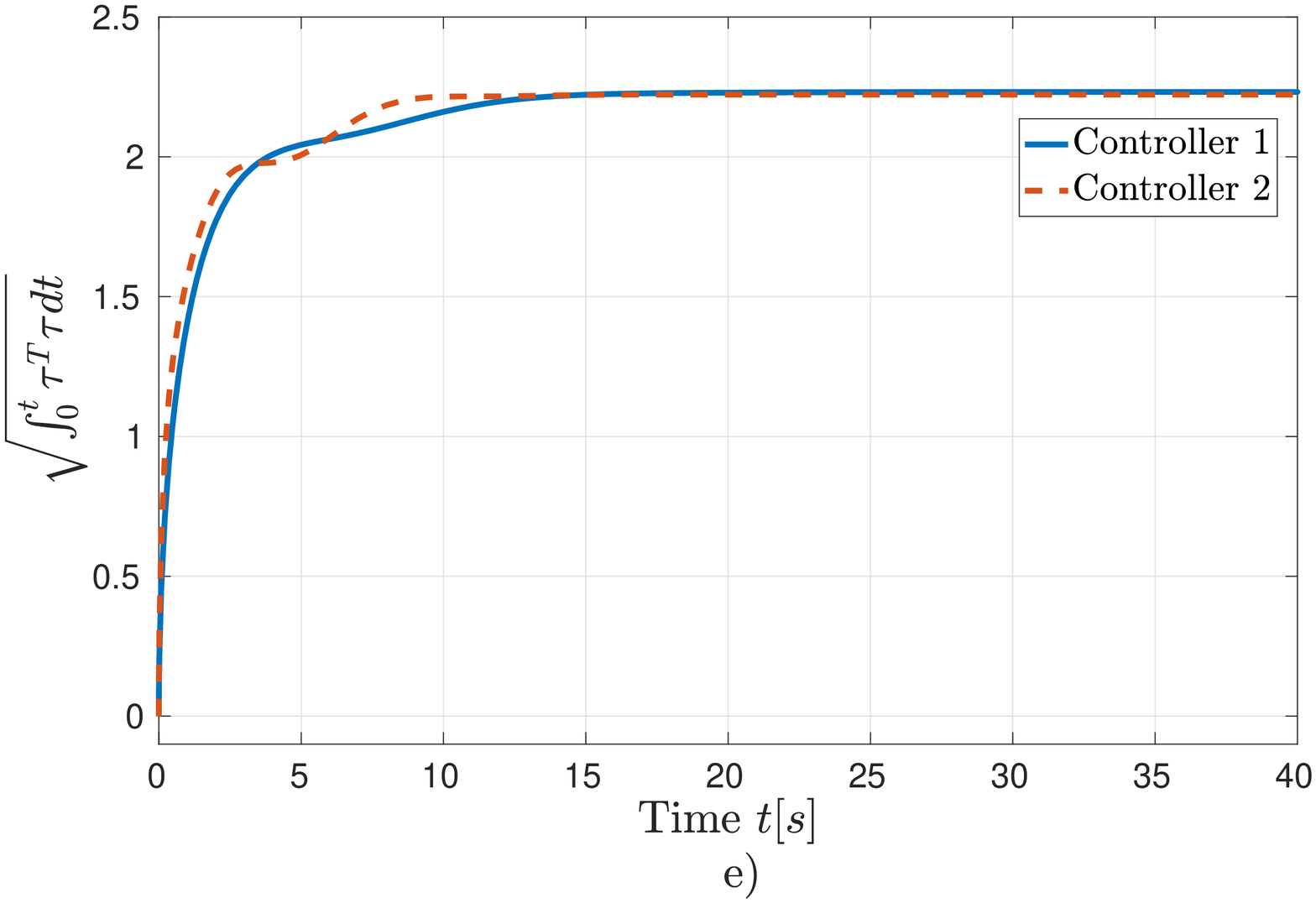}
		\caption{Performances of proposed controller (controller 1) for $\delta = 1$ and controller reported in  \cite{thienel2003coupled} (controller 2) under the  same profile of energy.}
		\label{fig:ID_Graf}
	\end{center}
\end{figure}

\begin{figure}
	\begin{center}
		\includegraphics[scale=0.3]{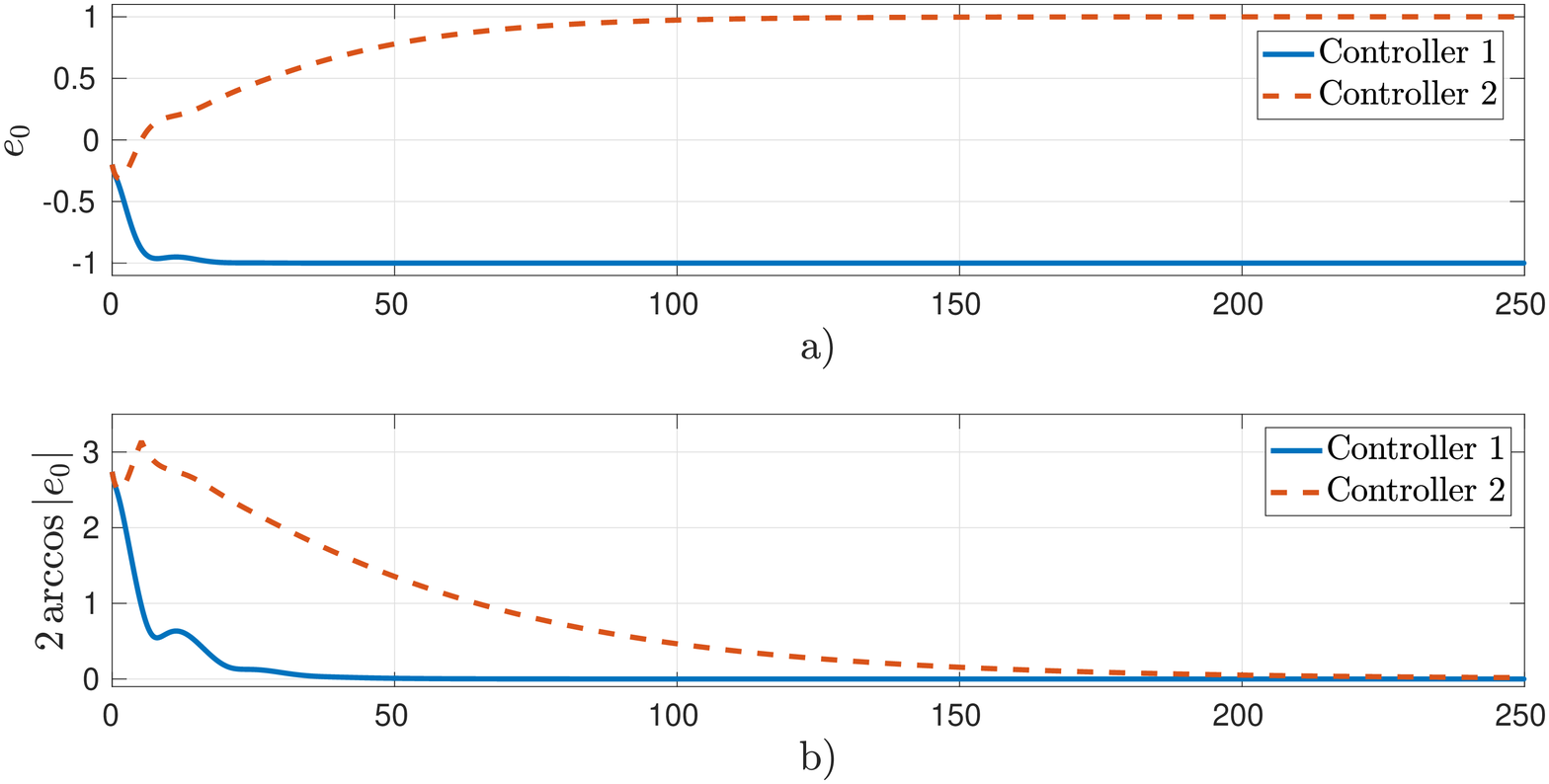}
		\includegraphics[scale=0.3]{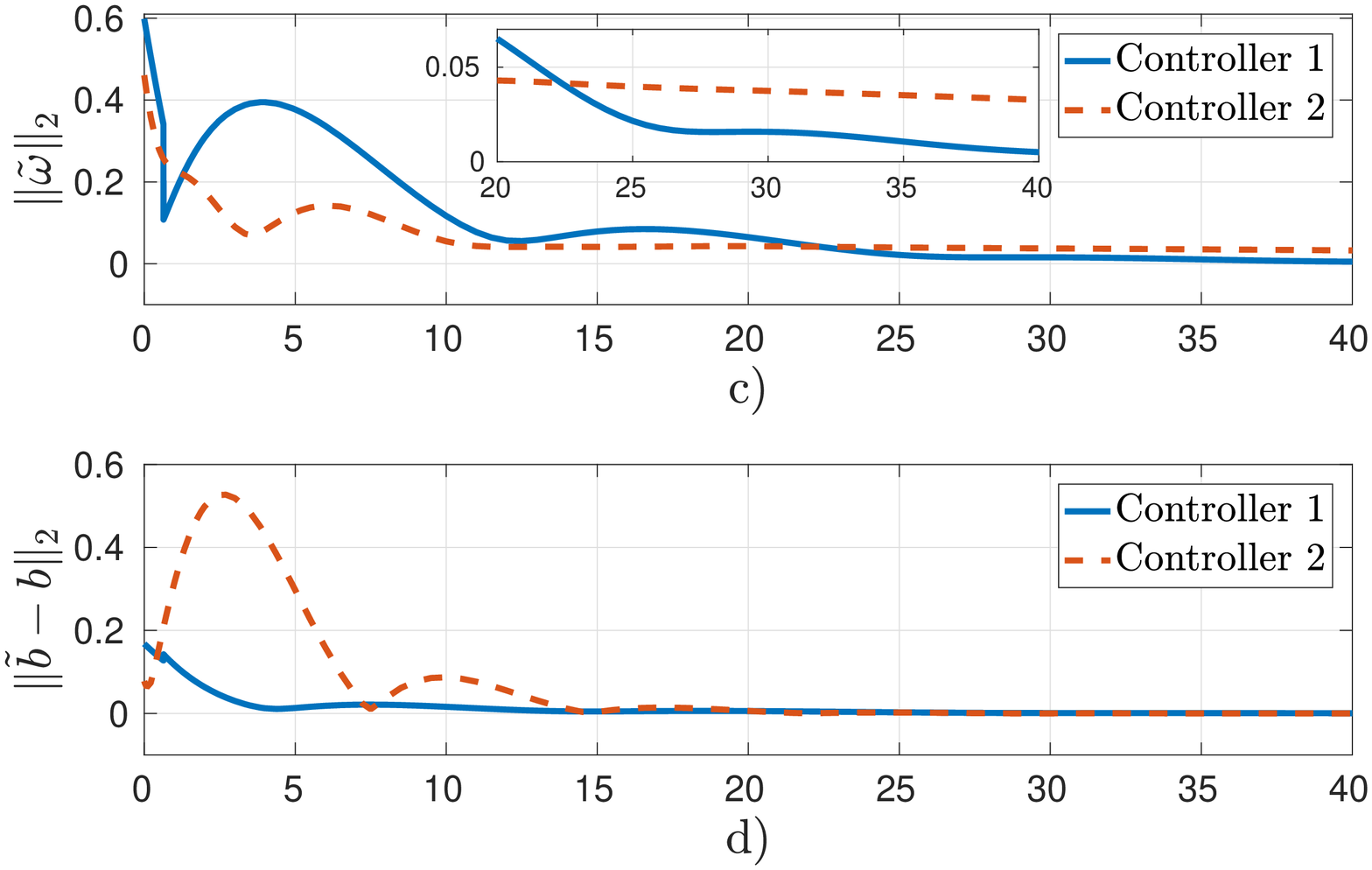}
		\includegraphics[scale=0.3]{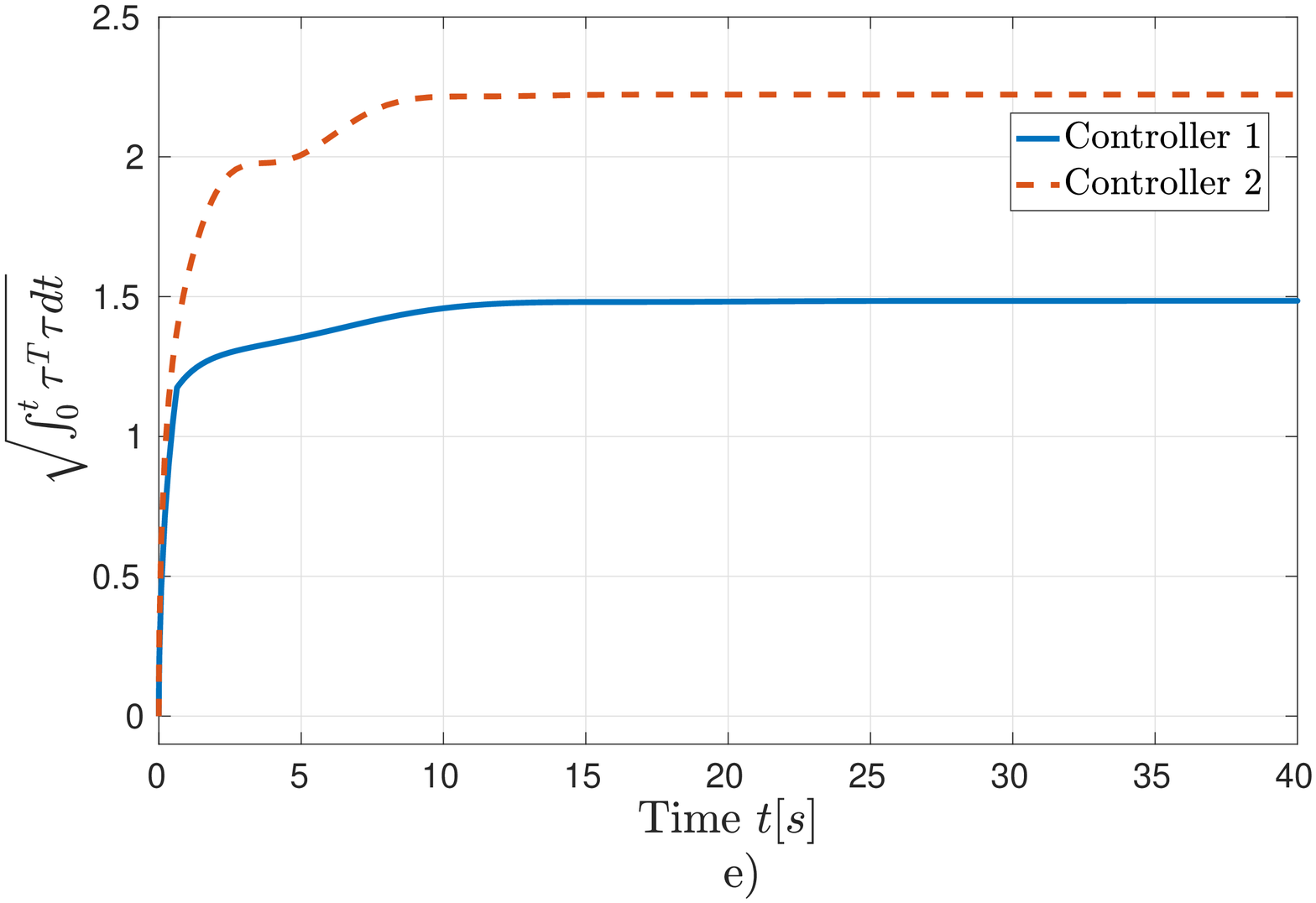}
		\caption{Behavior of the proposed controller (Controller 1) with $\delta = 0.3$ and the controller reported in  \cite{thienel2003coupled} (controller 2) when unwinding phenomenon was presented.}
		\label{fig:UW_Graf}
	\end{center}
\end{figure}

\begin{figure}
	\begin{center}
		\includegraphics[scale=0.3]{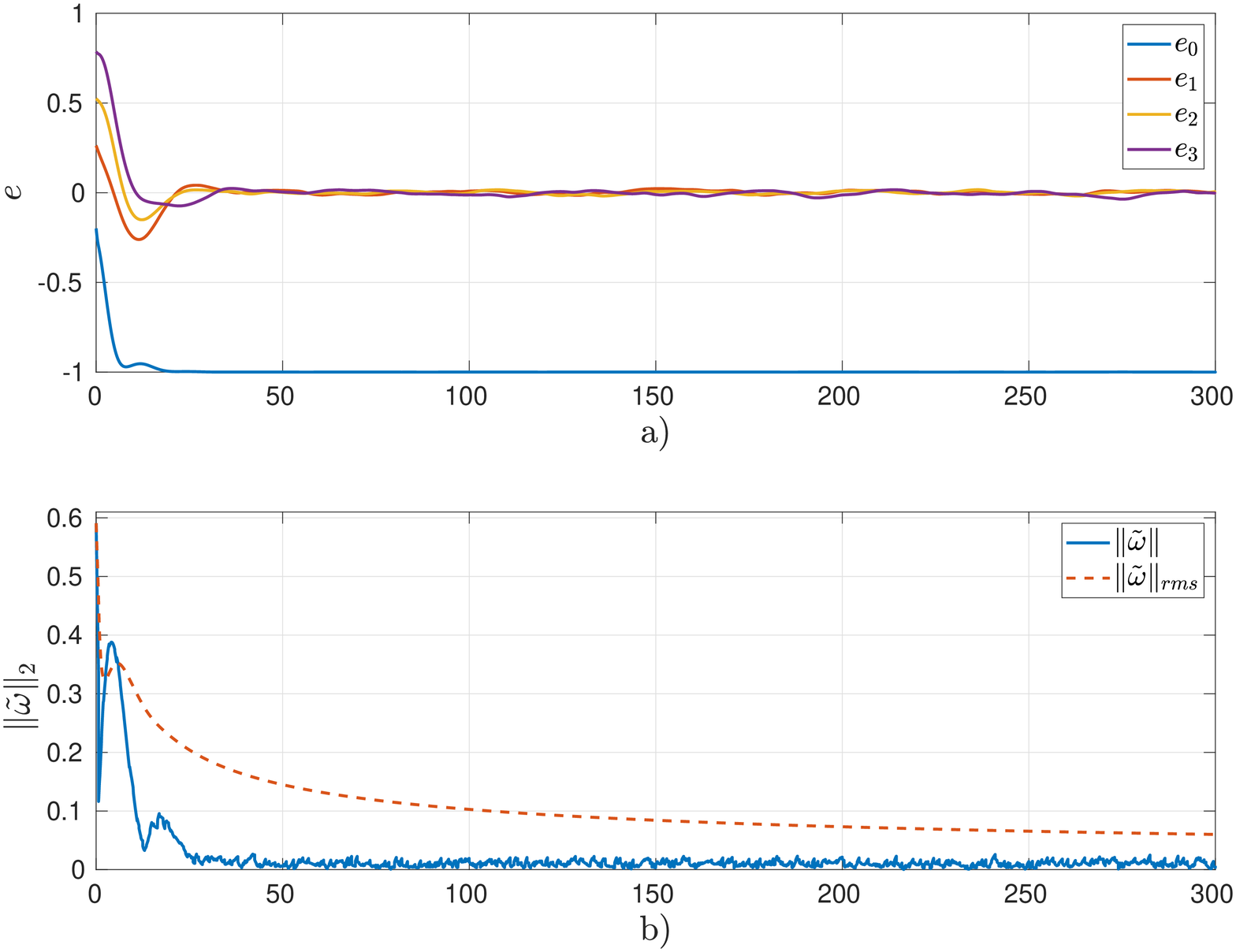}
		\includegraphics[scale=0.3]{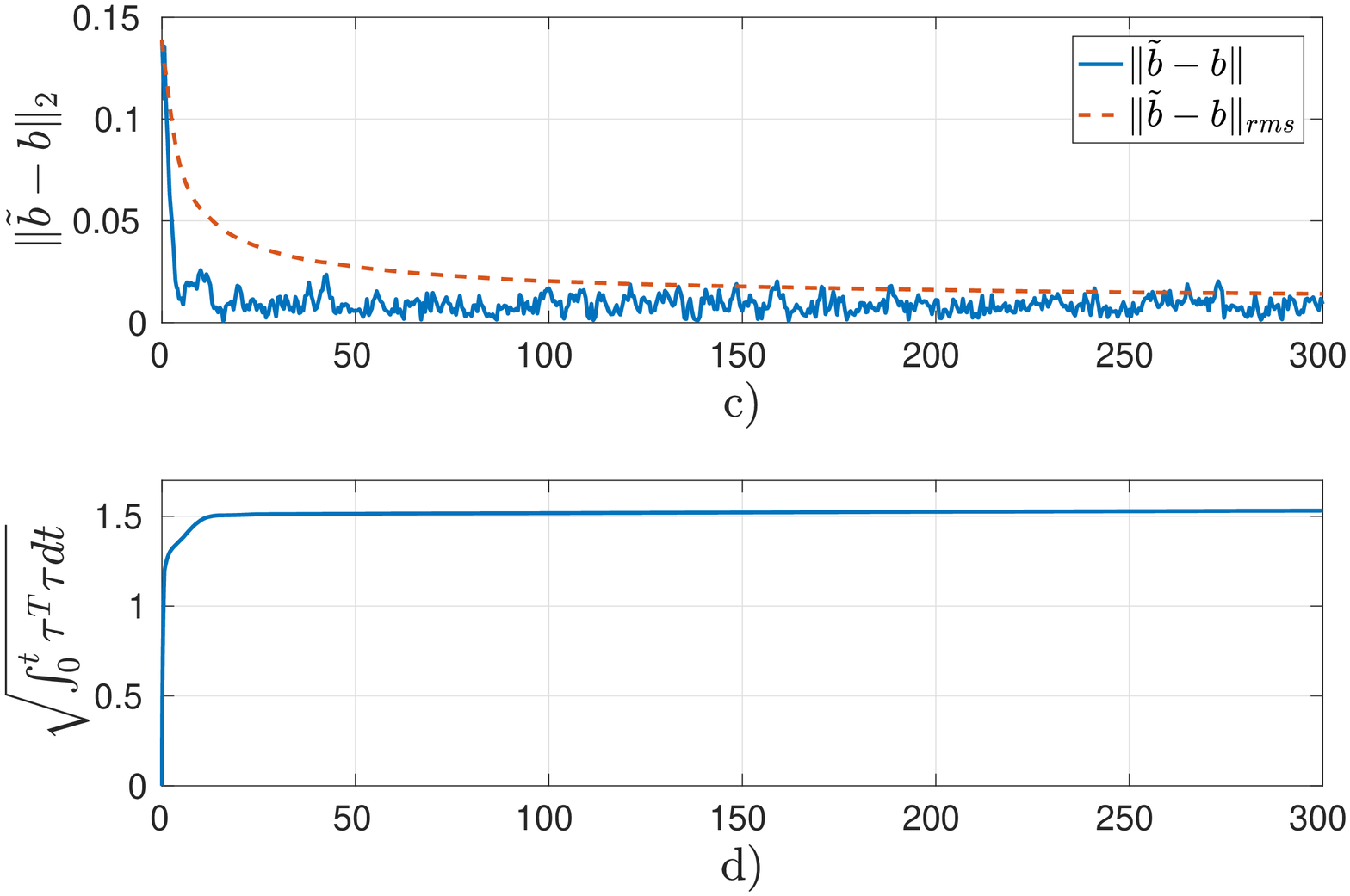}
		\includegraphics[scale=0.3]{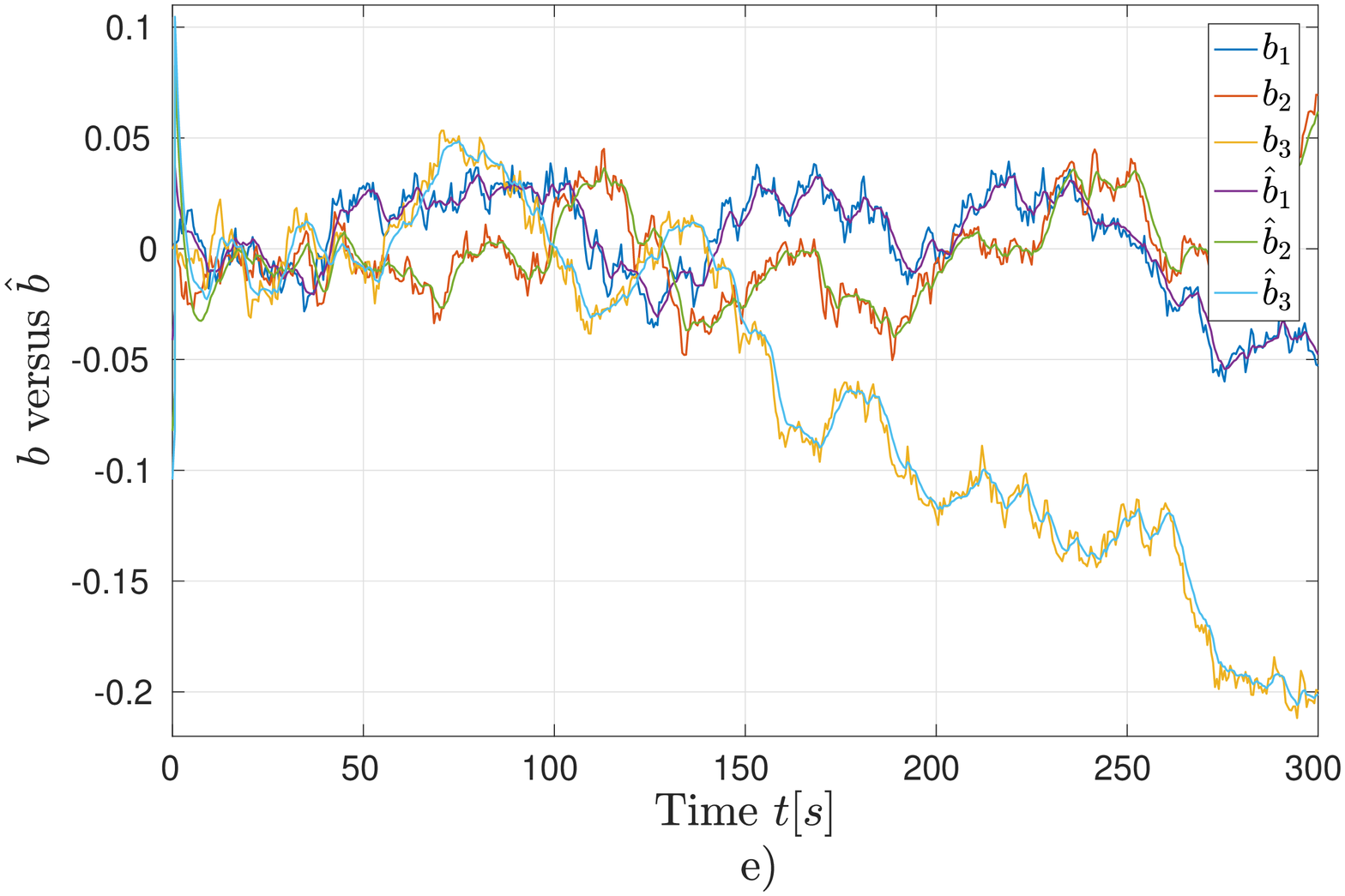}
		\caption{Performance of the proposed  controller  \eqref{eq:OFeedB_UF}-\eqref{eq:OmRpE_UF} for $\delta = 0.3$ in the presence of noisy gyro measurements and time-varying bias.}
		\label{fig:R_Graf}
	\end{center}
\end{figure}

\section{Conclusions}\label{Sec:Conc}
This paper has presented a global exponential attitude controller with gyro bias estimation for a spacecraft using measurements from a low-cost gyroscope. By using contraction analysis, a nonlinear gyro bias observer was designed with exponential convergence. This observer was modified to develop a controller  with global exponential convergence. The global result was achieved by means of representing the tracking error by its quaternion logarithm. Furthermore, by incorporating  a hysteretically switching variable into  in the previous controller, a unwinding-free exponentially convergent switching controller with gyro bias estimation was developed to address energy-efficiency and robustness with respect to the measurement noise. 
Simulations were carried out to illustrate the main features of the proposed controller.

The main disadvantage of the proposed controller is that it has no reduction property, i.e. in the regulation case ($\omega_d =0$) the control law does not reduce to a simple quaternion feedback. In addition, parameters knowledge is necessary to implement the control law. These issues will be considered for future investigation along this line.

\section*{Acknowledgment}                             
This work was supported in part by CONACyT under grant 253677 and by PAPIIT-UNAM IN113418, and carried out in the National Laboratory of Automobile and Aerospace Engineering LN-INGEA.  

\bibliographystyle{elsarticle-num} 

\bibliography{arxiv1}


%
%
%
%

\end{document}